\documentclass[a4paper,12pt, leqno]{article}

\usepackage{natbib,enumerate,graphicx,array}
\usepackage{amsmath}
\usepackage{amsthm}
\usepackage{amssymb}
\usepackage{hyperref}
\usepackage{amsfonts}
\usepackage{amscd}
\usepackage[utf8]{inputenc}
\usepackage{palatino}

\newtheorem{theorem}{Theorem}

\newtheorem{definition}[theorem]{Definition}

\newtheorem{lemma}[theorem]{Lemma}

\theoremstyle{definition}

\theoremstyle{remark}
\newtheorem{example}[theorem]{Example}

\newcommand* \diag {\mathop{\rm diag}\nolimits }

\newcommand* \ca {\mathop{\rm ca}\nolimits }
\newcommand* \cba {\mathop{\rm ca}\nolimits }

\newcommand* \cone {\mathop{\rm cone}\nolimits }
\newcommand* \conv {\mathop{\rm conv}\nolimits }

\newcommand*  \dd {\mathop{\ \rm d}\nolimits}

\newcolumntype{x}[1]{>{\centering\arraybackslash}p{#1}}
\newcolumntype{C}[1]{>{\centering\let\newline\\\arraybackslash\hspace{0pt}}m{#1}}
\usepackage{setspace}

\def\R{\mathbb{R}}
\def\N{\mathbb{N}}
\def\A{\mathcal{A}}
\def\M{\mathcal{M}}
\def\O{\Omega}
\def\o{\omega}

\def\:={\mathrel{\mathop:}=}

\title{Consistent Beliefs without Common Prior\thanks{We thank Jack Stecher and participants at SAET 2022, SOR 2023, the 14th Annual Financial Market Liquidity Conference, EWET 2024, RUD 2025, and SAET 2025 for their helpful comments and suggestions. Hellman gratefully acknowledges the support by Israel Science Foundation grants 448/22 and 458/25. Pintér acknowledges the support by the Hungarian Scientific Research Fund under the project K 146649 and by the HUN-REN under the project TKCS-2024/34.}}

\author{Ziv Hellman\thanks{Department of Economics, Bar-Ilan University, ziv.hellman@biu.ac.il.} \, and Mikl\'os Pint\'er\thanks{Corresponding author: Corvinus University of Budapest and HUN-REN-BME-BCE Quantum Technology Research Group, pmiklos@protonmail.com.}}

\begin{document}
	
\maketitle
	
\begin{abstract}
In a strand of the literature, it is assumed that the common prior has full support; that is, every type of every player is assigned positive probability. \citet{Morris1991,Morris1994} established an epistemological–behavioral duality characterisation of the common prior with full support, showing that a finite type space admits such a prior if and only if it contains no acceptable bet. This result forms the basis of the present paper.

The paper makes three contributions: (1) The characterisation of \citet{Morris1991,Morris1994} is extended to infinite type spaces. (2) The extension is robust: it does not depend on whether the infinite model applies countably additive or purely additive probabilities as beliefs. (3) The analysis implies that the notion of a real common prior—understood as a single probability distribution or a set of probability distributions—is not necessarily meaningful. 
		
\bigskip
		
\textit{Keywords:} Type space, Information structure, Common prior, Consistencies of Beliefs, Prior, Bet
\end{abstract}

\section{Introduction}

There are two strands in the literature concerning the notion of a common prior. The first treats the common prior as a probability distribution over the state space that is consistent with, or serves as a prior for, the posteriors of all players (see \citet{Aumann1976,Samet1998,Feinberg2000}, among others).

The second strand adopts the same notion of the common prior as does the one in the previous paragraph, but further assumes that the prior has full support (see \citet{Morris1991,Morris1994,Morris2020,FrickIijimalshii2023,MorrisOyamaTakahashi2024}, among others). Here, full support means that every type of every player is assigned positive probability under the common prior. Henceforth, we refer to this as the \textit{common prior with full support} assumption.

At first glance, the distinction between a common prior and a common prior with full support may appear to be minor, but we will show here that this is not the case. 
In the epistemic–behavioral dual characterisation \citep{Samet1998,Morris1994}\footnote{The epistemic pillar is the notion of a common prior, and the behavioural pillar is the notion of a bet (or trade). Characterisation means the pillars are equivalent.}, the behavioural counterpart of a common prior is an \textit{agreeable bet} \citep{Samet1998,Feinberg2000}, where every player at every state expects a positive payoff. In contrast, the behavioural counterpart of a common prior with full support is an \textit{acceptable bet} \citep{Morris1994}, where all expected payoffs are non-negative and some are strictly positive.

While the concept of full support is meaningful in finite state spaces, it is less so in infinite ones. Our first result generalises the notion of common prior with full support to infinite state spaces via the epistemic–behavioral dual characterisation: we extend acceptable bets to infinite state spaces and prove a corresponding dual characterisation theorem (Theorem \ref{thm:notrade3}). The resulting generalised notion of a common prior we term \textit{strong consistency of beliefs} (Definition \ref{def:cp3}).

When generalising the notion of a common prior to infinite state spaces, a key question is which model to adopt. One approach treats beliefs and (common) priors as probability measures ($\sigma$-additive probabilities) \citep{MertensZamir1985,BrandenburgerDekel1993,Heifetz1993,Morris1994,HeifetzSamet1998,Feinberg2000,LehrerSamet2014}.
There are other approaches that models beliefs and priors as probability charges (additive probabilities) \citep{HellmanPinter2022}.

The common prior–agreeable bet dual characterisation fails in the probability measures model \citep{Feinberg2000}, unless one adds significant additional assumptions such as compactness or a countable state space \citep{Feinberg2000,LehrerSamet2014}.
However, but it generalises straightforwardly in the probability charges model \citep{HellmanPinter2022}.

The second result of this paper is that the generalisation of strong consistency of beliefs (common prior with full support) — the acceptable bet epistemic–behavioral dual characterisation — is robust to the choice of model. That is, it holds both for probability measures (Theorem \ref{thm:notrade3}) and for probability charges (Supplementary Material). 
It follows that the strong consistency of beliefs characterisation is less fragile than the common prior–agreeable bet dual characterisation -- a point in favour of the former.

With regards to the  concept of type spaces, these admit two interpretations. One views them as models of asymmetric information, with a concrete ex-ante stage from which players update their beliefs. The other treats them as tools for incomplete information, representing hierarchies of beliefs while the ex-ante stage is merely an abstraction.

In the asymmetric information case, priors are elicited via ex-ante bets and have clear meaning. In the incomplete information case, priors are purely analytical, and the ex-ante stage is at most a fiction.

Regardless of whether the ex-ante stage is regarded as concrete or merely fictional, the common prior—as a probability distribution—serves as a witness that players’ beliefs are consistent in the type space. Mathematically, an epistemic–behavioural dual characterisation relies on the separation of convex sets: strong separation corresponds to the common prior–agreeable bet characterisation, while proper separation corresponds to the strong consistency of beliefs (common prior with full support)–acceptable bet characterisation. Hence, a common prior, interpreted as a probability distribution or a set of distributions, can indicate the presence or absence of strong or proper separation.

The third result of this paper is illustrated in Examples \ref{exp:Feinbergv1} and \ref{exp:Feinbergv2}. These examples show that proper separation of convex sets may not be detected by the intersection of the sets. In finite-dimensional spaces, two convex sets are properly separated if and only if the intersection of their relative interiors is nonempty; an element of this intersection serves as a witness to proper separation. Each example features two players and two closed convex sets in the space of probability distributions with identical intersections. In one example, the sets are properly separable; in the other, they are not.

These examples imply that, in infinite state spaces, the common prior with full support cannot generally be represented by a probability distribution or a set of distributions. Consequently, not only may a concrete ex-ante stage be meaningless, but even a probability distribution as a convenient abstraction may fail to capture the notion. Building on this observation, we propose adopting the term (strong) consistency of beliefs instead of common prior (with full support). This result contributes to the debate initiated by \citet{Aumann1998} and \citet{Gul1998}.   

In summary, given the insights gained from infinite models, we prefer the term strong consistency of beliefs for  common prior with full support, and more generally advocate the use of consistency of beliefs over the term common prior.

\subsection{Outline of the Paper}

The structure of the paper is as follows. 
In the main section (Section \ref{sec:main}), we introduce the model, presenting the definitions of type space, bet, and acceptable bet, as well as a characterisation of acceptable bets using proper separation. This is followed by a presentation of the notion of strong consistency of beliefs (Definition \ref{def:cp3}).

In finite state spaces, there is a simple characterisation of strong consistency: beliefs there are strongly consistent if and only if the intersection of the relative interiors of the players' prior sets is non-empty, which essentially means that strong consistency can be witnessed by a probability distribution.
In general state spaces, this characterisation no longer holds. 
Instead, we present an epistemic–behavioral dual characterisation (Theorem \ref{thm:notrade3}), using the notion of acceptable bets.

Section \ref{sec:exp} presents the examples discussed above. The final section provides a brief conclusion. In Appendix \ref{app}, we collect the technical results needed for the paper, while the Supplementary Material considers these results in the context of a model with probability charges.

\section{Main result}\label{sec:main}

We begin by introducing the notion of type space that we use in this paper. This notion is a reformulation of that of \citet{HeifetzSamet1998}, which is considered the standard concept of type space.

\begin{definition}\label{def:ts}
A type space is a tuple $((\Omega,\mathcal{A}),\{(\Omega,\mathcal{M}_i)\}_{i \in N},\{t_i\}_{i \in N})$, where

\begin{itemize}
\item $N$ is the nonempty set of players,

\item $\Omega$ is the set of states of the world (the state space),

\item $\mathcal{M}_i$ is the class of events known to player $i$, called player $i$'s knowledge $\sigma$-field; it is a $\sigma$-field over $\Omega$ for each $i \in N$,

\item $\mathcal{A}$ is the $\sigma$-field of all considered events, satisfying $\mathcal{M}_i \subseteq \mathcal{A}$ for all $i \in N$,

\item for each player $i \in N$, $t_i \colon \Omega \times \mathcal{A} \to [0,1]$ is a type function satisfying:

\begin{enumerate}
\item $t_i(\omega, \cdot) \in \Delta(\Omega, \mathcal{A})$ represents the belief of player $i$ at state $\omega$, where $\Delta$ denotes the set of probability measures, for all $\omega \in \Omega$,

\item $t_i(\cdot, E)$ is $\mathcal{M}_i$-measurable; that is, player $i$ knows her own belief about the event $E$, for all $E \in \mathcal{A}$,

\item for each $E \in \mathcal{M}_i$ and each $\omega \in E$, $t_i(\omega,E) = 1$; that is, player $i$ assigns probability 1 to events she knows.
\end{enumerate}
			
\end{itemize}
\end{definition}

The behavioural pillar of our analysis is the notion of bet (also called trade, see e.g. \citet{Morris1994}):

\begin{definition}\label{def:trade1}
A \emph{bet} $f=(f_{i})_{i \in I}$, where $I \subseteq N$, $| I | < \infty$, is a family of functions, one for each player $i \in I$, such that $f_i$ is bounded and $\mathcal{A}$-measurable for all $i$, and furthermore satisfies

\begin{equation*}
\sum_{i \in I} f_i(\omega)  = 0 \, 
\end{equation*}

\noindent for every $\o \in \O$.
\end{definition}

The notion of acceptable bet (called acceptable trade in \citet{Morris1994}) can be generalised from the finite state space setting to the infinite state space setting as follows:

\begin{definition}\label{def:trade5}
A bet $f=(f_{i})_{i \in I}$ is \emph{acceptable} if 
\begin{equation*}
\int f_i \dd t_i (\o,\cdot) \geq 0 \mspace{15mu} 
\end{equation*}
\noindent \text{for all $i \in I$ at every state of the world $\o \in \O$}, and there exists a state of the world $\o^\ast \in \O$ and player $i^\ast \in I$ such that 
\begin{equation*}
\int f_{i^\ast} \dd t_{i\ast} (\o^\ast,\cdot) > 0 \, .
\end{equation*}
\end{definition}
	
In words, a bet is acceptable if, at every state of the world, it is commonly certain that no player incurs a loss, while there exists at least one state at which some player expects a strictly positive payoff.
	
We can reformulate the notion of acceptable bet (Definition \ref{def:trade5}) as follows:
	
\begin{lemma}\label{def:trade6}
A bet $f=(f_{i})_{i \in I}$ is \emph{acceptable} if and only if there exists $i^\ast \in I$ such that $(f_i)_{i \in I \setminus \{i^\ast\}}$ properly separates the sets $\diag \Pi_{i^\ast}^{|I|-1}$ and $\cone (\prod_{i \in I \setminus \{i^\ast\}} \Pi_i)$, where
\begin{equation*}
\Pi_i = \overline{ \left \{ \mu \in \Delta (\O,\A) \colon \mu (E \cap F) = \int \limits_F t_i (\cdot,E) \dd \mu,\ \forall E \in \A,\ F \in \mathcal{M}_i \right \}}^\ast \, 
\end{equation*}

\noindent for all $i \in I$, here $^\ast$ denotes the weak* closure, $\diag$ denotes the diagonal and $\cone$ denotes the cone hull. 

$\Pi_i$ is called \emph{the set of player $i$'s priors}. 
\end{lemma}

As a reminder, here we give the definitions of proper and strong separation of convex set. Given two convex sets $C_1,C_2 \subseteq X$. Then a linear functional $f \colon X \to \R$ properly separates $C_1$ and $C_2$, if there exists $c \in \R$ such that $f (x) \geq c$ for all $x \in C_1$, $f (x) \leq c$ for all $x \in C_2$ and $f$ is not $0$ over $C_1 \cup C_2$, that is, there exists $x \in C_1 \cup C_2$ such that $f (x) \neq 0$. Furthermore, $f$ strongly separates $C_1$ and $C_2$, if there exist $c \in \R$ and $\varepsilon > 0$ such that $f (x) \geq c + \varepsilon$ for all $x \in C_1$, and  $f (x) \leq c - \varepsilon$ for all $x \in C_2$.
	
\begin{proof}
If: Suppose that $(f_i)_{i \in I \setminus \{i^\ast\}}$, where $(f_i)_{i \in I \setminus \{i^\ast\}}$ is the product of the bets $f_i$, $i \in I \setminus \{i^\ast\}$, properly separates the sets $\diag \Pi_{i^\ast}^{|I|-1}$ and $\cone (\prod_{i \in I \setminus \{i^\ast\}} \Pi_i)$ for some $i^\ast \in I$. Then without loss of generality, $\sum_{i \in I \setminus \{i^\ast\}} f_i (x_i) \geq 0$ for every $(x_i)_{i \in I \setminus \{i^\ast\}} \in \cone (\prod_{i \in I \setminus \{i^\ast\}} \Pi_i)$, where $ f_i (x_i) = \int f_i \dd x_i$ ($x_i$ is a probability measure), and $\sum_{i \in I \setminus \{i^\ast\}} f_i (x) \leq 0$ for every $x \in \Pi_{i^\ast}$, and there exists $(x_i)_{i \in I \setminus \{i^\ast\}} \in \diag \Pi_{i^\ast}^{|I|-1} \cup \, \cone (\prod_{i \in I \setminus \{i^\ast\}} \Pi_i)$ such that $\sum_{i \in I \setminus \{i^\ast\}} f_i (x_i) \neq 0$.
		
Since $(x_i,0_{j \in I \setminus \{i^\ast,i\}}) \in \cone (\prod_{j \in I \setminus \{i^\ast\}} \Pi_j)$ for each $i \in I \setminus \{i^\ast\}$ and $(x_j)_{j \in I \setminus \{i^\ast\}} \in \cone (\prod_{j \in I \setminus \{i^\ast\}} \Pi_j)$, we have that $f_i (x_i) \geq 0$ for every $i \in I \setminus \{i^\ast\}$ and $x_i \in \Pi_i$, that is, $\int f_i \dd t_i (\o,\cdot) \geq 0$ for every $i \in I \setminus \{i^\ast\}$ and $\o \in \O$. 
		
Moreover, let $f_{i^\ast} = - \sum_{i \in I \setminus \{i^\ast\}} f_i$. 
Then $f_{i^\ast} (x_{i^\ast}) \geq 0$ for every $x_{i^\ast} \in \Pi_{i^\ast}$, that is, $\int f_{i^\ast} \dd t_{i^\ast} (\o,\cdot) \geq 0$ for every $\o \in \O$. Furthermore, there exists $\o \in \O$ and $i \in I$ such that $\int f_{i} \dd t_{i} (\o,\cdot) > 0$.
		
\bigskip
		
Only if:  Suppose that $(f_i)_{i \in I}$ is an acceptable bet, and let $i^\ast \in I$ be such that $\int f_{i^\ast} \dd t_{i^\ast} (\o,\cdot) > 0$ for some $\o \in \O$. 
Then $\int f_i \dd t_i (\o,\cdot) \geq 0$ for every $i \in I$ and $\o \in \O$, that is, $f_i (x_i) \geq 0$ for every $i \in I$ and $x_i \in \Pi_i$. 

Therefore, $(f_i)_{i \in I \setminus \{i^\ast\}} ((x_i)_{i \in I \setminus \{i^\ast\}}) = \sum_{i \in I \setminus \{i^\ast\}} f_i (x_i) \geq 0$ for every $(x_i)_{i \in I \setminus \{i^\ast\}} \in \cone (\prod_{i \in I \setminus \{i^\ast\}} \Pi_i)$; and $(f_i)_{i \in I \setminus \{i^\ast\}} (x) = \sum_{i \in I \setminus \{i^\ast\}} f_i (x_i) = - f_{i^\ast} (x_j)  \leq 0$ for every $x \in \diag \Pi_{i^\ast}^{|I|-1}$ and $j \in I \setminus \{i^\ast\}$. 

Moreover, there exists $x_{i^\ast} \in \Pi_{i^\ast}$ such that $\sum_{i \in I \setminus \{i^\ast\}} f_i (x_{i^\ast}) = - f_{i^\ast} (x_i) \neq 0$; meaning there exists $(x_i)_{i \in I \setminus \{i^\ast\}} \in \diag \Pi_{i^\ast}^{|I|-1} \cup \, \cone (\prod_{i \in I \setminus \{i^\ast\}} \Pi_i)$ such that $(f_i)_{i \in I \setminus \{i^\ast\}} ((x_i)_{i \in I \setminus \{i^\ast\}}) \neq 0$. Therefore, $(f_i)_{i \in I \setminus \{i^\ast\}}$ properly separates the sets $\diag \Pi_{i^\ast}^{|I|-1}$ and $\cone (\prod_{i \in I \setminus \{i^\ast\}} \Pi_i)$.
\end{proof}

The concept of common prior with full support is problematic to define in infinite spaces.
We therefore present the following notion, strong consistency of beliefs, as our suggested infinite state space generalisation of the notion of common prior with full support:

\begin{definition}\label{def:cp3}
A set of players' beliefs in a type space $((\Omega,\mathcal{A}),\{(\Omega,\mathcal{M}_i)\}_{i \in N},\{t_i\}_{i \in N})$ are \emph{strongly consistent} if for every $I \subseteq N$, $|I| < \infty$, every player $i^\ast \in I$, every type $t_{i^\ast} (\o^\ast , \cdot)$ of player $i^\ast$, and every state of the world $\o^\ast \in \O$: 
\begin{equation*}
t_{i^\ast} (\o^\ast,\cdot) \in \overline{\cone (\prod_{i \in I \setminus \{i^\ast\}} \Pi_i) - \cone(\diag (\Pi_{i^\ast}^{-\o^\ast})^{|I| -1})}^\ast \,
\end{equation*}

%\begin{equation*}
%t_{i^\ast} (\o^\ast,\cdot) \in \overline{\cone (\{t_i (\o, \cdot) \colon \o \in \O \} - \{t_{i^\ast} (\o,\cdot) \colon \o \in \O \text{ s.t. } t_{i^\ast} (\o,\cdot) \neq t_{i^\ast} (\o^\ast,\cdot)  \})}^\ast \, ,
%\end{equation*}
		
\noindent for all $i \in I \setminus \{i^\ast\}$, where $\Pi_{i^\ast}^{-\o^\ast} = \overline{\{t_{i^\ast} (\o,\cdot) \colon \o \in \O \text{ s.t. } t_{i^\ast} (\o,\cdot) \neq t_{i^\ast} (\o^\ast,\cdot)  \})}^\ast$.
\end{definition}
	
Two remarks apply here. First, in a finite type space (finite state space setting) the players' beliefs are strongly consistent if and only if the intersection of the relative interiors of  the players' prior sets ($\Pi_i$) is nonempty. In other words, in the finite type space setting the strong consistency of the players' beliefs can be witnessed by a probability distribution over the state space.

Second, the notion of strong consistency of beliefs looks troubling: it is very strange and, at first glance, counterintuitive. However, this is unavoidable—strong consistency of beliefs must look troubling and must lack an intuitive common-prior like interpretation. As we will see in Examples \ref{exp:Feinbergv1} and \ref{exp:Feinbergv2}, strong consistency of beliefs generally cannot be described by a prior shared by the players. Instead, strong consistency of beliefs is characterized by the relative positions of the players’ priors, which is a technical feature. In other words, the common view that the consistency of players’ beliefs can be understood as arising from a (common) prior at an ex-ante stage, or even as being captured through aggregations of the players’ beliefs (their priors), cannot be sustained.

\begin{lemma}\label{lem:important0}
The players' beliefs in a finite type space $((\Omega,\mathcal{A}),\{(\Omega,\mathcal{M}_i)\}_{i \in N},\{t_i\}_{i \in N})$ are strongly consistent if and only if $\cap_{i \in N} \textup{int} \Pi_i \neq \emptyset$, where $\textup{int}$ denotes the relative interior.
\end{lemma}

\begin{proof}
We consider the two-player case; the proof of the general case follows analogously.

Let $N = \{1,2\}$ and $\mu \in \textup{int} \Pi_1 \cap \textup{int} \Pi_2$. Then there exist $n^1,n^2 \in \N$, and $(\alpha_{m^i}^i) \subseteq \R_{++}$, $m^i = 1,\ldots,n^i$ such that $\sum_{m^i=1}^{n^i} \alpha_{m^i}^i = 1$, $i=1,2$, and 

\begin{equation*}
\sum_{m^1=1}^{n^1} \alpha_{m^1}^1 t_1^{m^1} = \mu = \sum_{m^2=1}^{n^2} \alpha_{m^2}^2 t_2^{m^2} \, ,
\end{equation*}

\noindent where $t_i^{m^i} \in \{ t_i (\o,\cdot) \colon \o \in \O\}$, and $n^i = | \{ t_i (\o,\cdot) \colon \o \in \O\} |$, $i=1,2$.

Then for any $t_1^{m'^1}$ ($t_2^{m'^2}$), 
\begin{equation}\label{eq1}
\begin{split}
t_1^{m'^1} = \dfrac{1}{\alpha_{m'^1}^1} \left( \sum \limits_{m^2=1}^{n^2} \alpha_2^{m^2} t_2^{m^2}   - \sum \limits_{m^1 \in \{1,\ldots,n^1\} \setminus \{m'^1\}} \alpha_1^{m^1} t_1^{m^1} \right)  \\ 
\in \cone ((\{t_2^{m^2} \colon m^2 = 1,\ldots,n^2 \} - \{t_1^{m^1} \colon m^1 \in \{1,\ldots,n^1\} \setminus \{m'^1\} \})
\end{split}
\end{equation}

\begin{equation}\label{eq2}
\left(
\begin{split}
t_2^{m'^2} = \dfrac{1}{\alpha_{m'^2}^2} \left( \sum \limits_{m^1=1}^{n^1} \alpha_1^{m^1} t_1^{m^1}   - \sum \limits_{m^2 \in \{1,\ldots,n^2\} \setminus \{m'^2\}} \alpha_2^{m^2} t_2^{m^2} \right)  \\ 
\in \cone ((\{t_1^{m^1} \colon m^1 = 1,\ldots,n^1 \} - \{t_2^{m^2} \colon m^2 \in \{1,\ldots,n^2\} \setminus \{m'^2\} \})
\end{split}
\right)
\end{equation}

In words, the intersection of the relative interiors of $\Pi_1$ and $\Pi_2$ is nonempty if and only if  for any $t_1^{m'^1}$ ($t_2^{m'^2}$) \eqref{eq1} (\eqref{eq2}) holds; that is, if and only if the beliefs of the players are strongly consistent. 
\end{proof}

The characterisation of strong consistency in Lemma \ref{lem:important0} does not extend to infinite state spaces. We therefore work instead with an infinite state space generalisation of the concept of strong consistency of beliefs.
The following theorem presents the appropriate infinite state space generalisation of the strong consistency of beliefs–acceptable bet epistemic–behavioral dual characterisation of \citet{Morris1994}. It states that the players’ beliefs in a type space are strongly consistent if and only if there exists no acceptable bet within the type space.
	
\begin{theorem}\label{thm:notrade3}
Let $T = ((\O,\A),(\O,\M_i)_{i \in N},(t_i)_{i \in N})$ be a type space. Then only one of the following two statements can obtain: 
		
\begin{itemize}
\item The players' beliefs in the type space are strongly consistent.
			
\item There exists an acceptable bet within the type space.
\end{itemize}
\end{theorem}

\begin{proof}%[The proof of Theorem \ref{thm:notrade3}]
Suppose that the players' beliefs in $T$ are strongly consistent, and suppose by contradiction that $f = (f_i)_{i \in I}$ is an acceptable bet. Let $i^\ast \in I$ and $\o^\ast \in \O$ be such that $\int f_{i^\ast} \dd t_{i^\ast} (\o^\ast,\cdot) > 0$. Then $(f_i)_{i \in I \setminus \{i^\ast\}} \colon \ca (\O,\A)^{I \setminus \{i^\ast\}} \to \R$ properly separates $\overline{\cone(\diag \Pi_{i^\ast}^{| I | -1})}^\ast$ and $\overline{\cone(\prod_{i \in I \setminus \{i^\ast\}} \Pi_i)}^\ast$; and%. Therefore, by Theorem \ref{thm:sep} we have that

\begin{equation*}
t_{i^\ast} (\o^\ast,\cdot) \notin \overline{\cone (\prod_{i \in I \setminus \{i^\ast\}} \Pi_i) - \cone(\diag (\Pi_{i^\ast}^{-\o^\ast})^{|I| -1})}^\ast \, ,
\end{equation*}

\noindent which is a contradiction, where $\ca$ denotes the bounded countable additive set functions (bounded signed measures).
		
\bigskip
		
Suppose that there is no acceptable bet within a type space. By Lemma \ref{def:trade6} this means that $\diag \Pi_{i^\ast}^{|I| - 1}$ and $\cone(\prod_{i \in I \setminus \{i^\ast\}} \Pi_i)$ cannot be properly separated for every $I \subseteq N$, $| I | < \infty$, $i^\ast \in I$, that is, the sets $\overline{\cone(\diag \Pi_{i^\ast}^{|I| -1})}^\ast$ and $\overline{\cone(\prod_{i \in I \setminus \{i^\ast\}} \Pi_i)}^\ast$ cannot be properly separated either. 
Then, by Theorem \ref{thm:sep} we have that for every $I \subseteq N$, $|I| < \infty$, $i^\ast \in I$ and $\o^\ast \in \O$,
\begin{equation*}
t_{i^\ast} (\o^\ast,\cdot) \in \overline{\cone (\prod_{i \in I \setminus \{i^\ast\}} \Pi_i) - \cone(\diag (\Pi_{i^\ast}^{-\o^\ast}))}^\ast \, .
\end{equation*}
		
\noindent Therefore, the beliefs of the players are strongly consistent in the type space.
\end{proof}

\section{Examples}\label{sec:exp}

The following two examples illustrate that the intersection of two closed convex sets may fail to witness their proper separation. This observation highlights a subtle geometric limitation underlying the epistemic–behavioral dual characterisation: in infinite-dimensional settings, proper separation cannot, in general, be represented by a single probability distribution or even by a set of distributions lying in the intersection. Consequently, these examples suggest that there may be no truly ‘elegant’ a distribution-based definition of strong consistency of beliefs.

\begin{example}\label{exp:Feinbergv1}
Consider the following type space which is an uncountable version of an example due to Yossi Feinberg (Figure 2 on p. 152 in \citet{Feinberg2000}).
				
The type space $((\Omega,\mathcal{A}),\{(\Omega,\mathcal{M}_i)\}_{i \in N},\{t_i\}_{i \in N})$ is as follows:
		
\begin{itemize}
\item $N = \{\text{Anne},\text{Ben}\}$,
			
\item $\O =\N \times [0,1]$, and $(m,x) \leq (n,y)$ if $x < y$ or $x = y$ and $m \leq n$. Then $(\O,\leq)$ is (linearly) ordered set,
			
\item $\A = \{ A \in \mathcal{P} (\O) \colon A  \text{ is countable or } A^\complement \text{ is countable}\}$, that is, $\A$ is the so called countable, co-countable $\sigma$-algebra,
			
\item $\M_{\text{Anne}}$ is the $\sigma$-field generated by the partition $\{\{(0,x)\},\{(1,x),(2,x)\},$ $\{(3,x),(4,x)\},\ldots\}$, $x \in [0,1]$, and $\M_{\text{Ben}}$ is the $\sigma$-field generated by the partition $\{\{(0,x),(1,x)\},\{(2,x),(3,x)\},\{(4,x),(5,x)\},\ldots\}$, $x \in [0,1]$,
			
\item for each $x \in [0,1]$: 

$t_{\text{Anne}} ((m,x),\cdot)= 
\begin{cases}
\delta_{(0,x)} & \text{if } m=0, \\
\frac{2}{3} \delta_{(m,x)} + \frac{1}{3} \delta_{(m+1,x)} & \text{if } m =1 + \sum_{k=1}^n 2^{k+1},\ \text{for some } n \in \N, \\
\frac{2}{3} \delta_{(m-1,x)} + \frac{1}{3} \delta_{(m,x)} & \text{if } m = \sum_{k=1}^n 2^{k+1},\ \text{for some } n \in \N, \\
\frac{1}{2} \delta_{(m,x)} + \frac{1}{2} \delta_{(m+1,x)} & \text{if } m \neq 1 + \sum_{k=1}^n 2^{k+1},\ \text{for every } n \in \N, \\
\frac{1}{2} \delta_{(m-1,x)} + \frac{1}{2} \delta_{(m,x)} & \text{if } m \neq \sum_{k=1}^n 2^{k+1},\ \text{for every } n \in \N , \\
\end{cases}$ and $t_{\text{Ben}} ((m,x), \cdot) = 
\begin{cases}
\frac{1}{2} \delta_{(m,x)} + \frac{1}{2} \delta_{(m+1,x)} & \text{if } m = 2n, \text{for some } n \in \N, \\
\frac{1}{2} \delta_{(m-1,x)} + \frac{1}{2} \delta_{(m,x)} & \text{if } m = 2n+1,\ \text{for some } n \in \N , \\
\end{cases}$
\end{itemize}
		
In this type space $\Pi_{\text{Anne}} \cap \Pi_{\text{Ben}} = \{\nu\}$, where

%the players' beliefs are consistent, and there is only one witness of this consistency: the uniform distribution, that is, the probability charge
		
\begin{equation*}
\nu (A) = 
\begin{cases}
0 & \text{if } | A | \text{ is countable} , \\
1 & \text{otherwise}. 
\end{cases}
\end{equation*}
		
%Consider $t_{\text{Anne}} (\o_1,\cdot)$. By Theorem \ref{thm:sep} if the point $t_{\text{Anne}} (\o_1,\cdot)$ and the cone $\overline{\cone (\{t_{\text{Ben}} (\o,\cdot) \colon \o \in \O\} - \{t_{\text{Anne}} (\o,\cdot) \colon \o \in \O \setminus \{\o_1\} \})}^\ast$ can be separated then $\Pi_{\text{Anne}}$ and $\Pi_{\text{Ben}}$ can be properly separated. We therefore need to check whether $\Pi_{\text{Anne}}$ and $\Pi_{\text{Ben}}$ can be properly separated or not.
		
Suppose that $f$ properly separates $\Pi_{\text{Anne}}$ and $\Pi_{\text{Ben}}$. Without loss of generality we can assume that $f (0,0) = 1$. Then by that $\int f \dd t_{\text{Ben}} ((0,0),\cdot) \leq 0$ we have that $f (1,0) \leq -1$. Then by that $\int f \dd t_{\text{Anne}} ((1,0),\cdot) \geq 0$ we have that $f (2,0) \geq 2$, and so on. Therefore, we get $f (7,0) \geq 4$, $f (15,0) \geq 8$, and so on, that is, $f$ is not bounded, hence it cannot be a separator; meaning $\Pi_{\text{Anne}}$ and $\Pi_{\text{Ben}}$ cannot be properly separated. 

Therefore (see Theorem \ref{thm:sep}), 
\begin{equation*}
t_{\text{Anne}} (\o_1,\cdot) \in \overline{\cone (\{t_{\text{Ben}} (\o,\cdot) \colon \o \in \O\} - \{t_{\text{Anne}} (\o,\cdot) \colon \o \in \O \setminus \{\o_1\}\})}^\ast \, .
\end{equation*}
		
By very similar reasoning one can show that the same  holds for every type of each player, therefore, the players' beliefs are strongly consistent. 
\end{example}
	
In words, in the example above, the players’ beliefs are strongly consistent -- or, equivalently, we can conclude that there is no acceptable bet within the type space.

Next, we slightly modify the example as follows:
	
\begin{example}\label{exp:Feinbergv2}
Consider the following type space which is a variant of the type space considered in Example \ref{exp:Feinbergv1} (which is based on Figure 2 on p. 152 in \citet{Feinberg2000})\footnote{The knowledge $\sigma$-fields and Ben type functions are the same in both examples, while in this type space Anne's belief is the uniform distribution at every state of the world.} .
				
The type space $((\Omega,\mathcal{A}),\{(\Omega,\mathcal{M}_i)\}_{i \in N},\{t_i\}_{i \in N})$ is as follows:
\begin{itemize}
\item $N = \{\text{Anne},\text{Ben}\}$,
			
\item $\O =\N \times [0,1]$, and $(m,x) \leq (n,y)$ if $x < y$ or $x = y$ and $m \leq n$. Then $(\O,\leq)$ is (linearly) ordered set,
			
\item $\A = \{ A \in \mathcal{P} (\O) \colon A  \text{ is countable or } A^\complement \text{ is countable}\}$, that is, $\A$ is what is called a countable, co-countable algebra,
			
\item $\M_{\text{Anne}}$ is the $\sigma$-field generated by the partition $\{\{(0,x)\},\{(1,x),(2,x)\},$ $\{(3,x),(4,x)\},\ldots\}$, $x \in [0,1]$, and $\M_{\text{Ben}}$ is the $\sigma$-field generated by the partition $\{\{(0,x),(1,x)\},\{(2,x),(3,x)\},\{(4,x),(5,x)\},\ldots\}$, $x \in [0,1]$,

\noindent Therefore, so far this type space and the previous one (Example \ref{exp:Feinbergv1}) are the same. The difference comes in the beliefs:  
			
\item for each $x \in [0,1]$: 

$t_{\text{Anne}} ((m,x),\cdot)= 
\begin{cases}
\delta_{(0,x)} & \text{if } m=0, \\
\frac{1}{2} \delta_{(m,x)} + \frac{1}{2} \delta_{(m+1,x)} & \text{if } m = 2n+1, \text{for some } n \in \N, \\
\frac{1}{2} \delta_{(m-1,x)} + \frac{1}{2} \delta_{(m,x)} & \text{if } m = 2(n+1),\ \text{for some } n \in \N , \\
\end{cases}$ and $t_{\text{Ben}} ((m,x), \cdot) = 
\begin{cases}
\frac{1}{2} \delta_{(m,x)} + \frac{1}{2} \delta_{(m+1,x)} & \text{if } m = 2n, \text{for some } n \in \N, \\
\frac{1}{2} \delta_{(m-1,x)} + \frac{1}{2} \delta_{(m,x)} & \text{if } m = 2n+1,\ \text{for some } n \in \N , \\
\end{cases}$
\end{itemize}
		
In this type space, as in the previous example, $\Pi_{\text{Anne}} \cap \Pi_{\text{Ben}} = \{\nu\}$, where $\nu$ is the uniform distribution.
		
%As in the previous example consider $t_{\text{Anne}} (\o_1,\cdot)$. By Theorem \ref{thm:sep} if the point $t_{\text{Anne}} (\o_1,\cdot)$ and the cone $\overline{\cone (\{t_{\text{Ben}} (\o,\cdot) \colon \o \in \O\} - \{t_{\text{Anne}} (\o,\cdot) \colon \o \in \O \setminus \{\o_1\} \})}^\ast$ can be separated then $\Pi_{\text{Anne}}$ and $\Pi_{\text{Ben}}$ can be properly separated. We therefore need to check whether $\Pi_{\text{Anne}}$ and $\Pi_{\text{Ben}}$ can be properly separated or not.
		
Consider a function $f (m,0) = 
\begin{cases}
1 & \text{if } m \text{ is even}, \\
-1 & \text{if } m \text{ is odd}.
\end{cases}$, and $f(n,x) = 0$ for all $m \in \N$ and $x \in (0,1]$.

Then it is easy to check that $\int f \dd t_{\text{Anne}} ((m,x),\cdot) \geq 0$ and $\int f \dd t_{\text{Ben}} ((m,x),\cdot) \leq 0$ for all $m \in \N$ and $x \in [0,1]$, and $\int f \dd t_{\text{Anne}} ((0,0),\cdot) = 1$. Therefore, $f$ is bounded and $\A$-integrable, and it properly separates  $\Pi_{\text{Anne}}$ and $\Pi_{\text{Ben}}$. 

Moreover, 
\begin{equation*}
t_{\text{Anne}} ((0,0),\cdot) \notin \overline{\cone (\{t_{\text{Ben}} (\o,\cdot) \colon \o \in \O\} - \{t_{\text{Anne}} (\o,\cdot) \colon \o \in \O \setminus \{(0,0)\}\})}^\ast \, .
\end{equation*}
		
Therefore, the players' beliefs are not strongly consistent here. 
\end{example}

Mathematically, in Examples \ref{exp:Feinbergv1} and \ref{exp:Feinbergv2}, we consider two pairs of closed convex sets—indeed, one set is common to both pairs—such that the intersections of the pairs coincide, and the relative interior of each closed convex set is empty. However, in one case the two sets cannot be properly separated, while in the other they can.

From an economic perspective, these examples show that even when two players share exactly the same prior (a common prior) — namely, when the intersections of their prior sets coincide and consist solely of the uniform distribution — their beliefs may differ in terms of strong consistency. In Example \ref{exp:Feinbergv1} the players’ beliefs are strongly consistent, whereas in Example \ref{exp:Feinbergv2} they are not. 

\section{Conclusions}

This paper generalises the strong consistency of beliefs (common prior with full support)–acceptable bet epistemic–behavioural dual characterisation of \citet{Morris1994} to the infinite state space setting. The resulting theorem offers new insights to the interpretation of the ex-ante stage and the notion of a prior.

The findings suggest that, in general, the ex-ante stage is best viewed as a non-primitive element of the model. Accordingly, the epistemic counterpart of the absence of arbitrage (an acceptable bet) should be interpreted as the (strong) consistency of players’ beliefs, rather than as a real probability distribution—namely, a common prior (with full support).  

\appendix\label{app}	

\section{Notions and Notations}

A measurable space is a tuple $(X,\A)$, where $X$ is a set and $\A$ is a $\sigma$-field, that is, $\A$ is closed under countable union and complement, and it contains the empty set. $\Delta (X,\A)$ denotes the set of probability measures over the measurable space $(X,\A)$. A probability measure is a non-negative, $\sigma$-additive set function which assigns $1$ to the whole set $X$. 

A function $f \colon X \to \R$ is measurable if  $f^{-1} (A) \in \A$ for every Borel set $A \in B(\R)$. A function $f \colon X \to \R$ is bounded if there exists $B \in \R$  such that $| f (x) | \leq B$ for all $x \in X$. Notice that for every bounded $\A$-measurable function $f$, and probability measure $\mu \in \Delta(X,\A)$ the integral $\int f \dd \mu$ exists (well-defined). Let $B(X,\A)$ denote the set of bounded $\A$-measurable functions. We consider the dual pair $(B(X,\A),\cba(X,\A))$, where $\cba(X,\A)$ denotes the set of bounded measures over $(X,\A)$.

The weak* topology on $\cba (X,\A)$ is defined as follows: for any $\mu \in \cba (X,\mathcal{A})$, $\varepsilon > 0$, and $F \subseteq B (X,\mathcal{A})$ such that $| F | < \infty$, denote by 

\begin{equation*}
\mathcal{O}^\ast (\mu,F,\varepsilon) = \left\{ \nu \in \cba (X,\mathcal{A}) \colon \left| \int f \dd \nu - \int f \dd \mu \right| < \varepsilon, \  \forall f \in F \right\},
\end{equation*} 
	
\noindent a neighbourhood of $\mu$. The collection of such neighbourhoods forms a subbase for the weak* topology on $\cba (X,\mathcal{A})$. $\overline{A}^\ast$ denotes the weak* closure of the set $A \subseteq \cba (X,\mathcal{A})$.
	
\section{Separation}
	
Given a dual pair $(X,X')$, let $^\ast$ denote the weak topology on $X$ (see e.g. Definition 5.90 on p. 211 in \citet{AliprantisBorder2006}). %Given a topology over the vector space $X$ the set of continuous linear functionals over $X$ is denoted by $X^\ast$.
	
\begin{lemma}\label{lem:appkell}
Let $X$ be a vector space, $E_1,E_2 \subseteq X$, and $f \in X'$ be such that $f(x) \geq 0 \geq f (y)$, for all $x \in\overline{\conv (E_1)}^\ast$, $y \in \overline{\conv (E_2)}^\ast$. If $x^\ast \in E_1$ is such that $x^\ast \in \overline{\cone (E_2 \cup - (E_1 \setminus \{ x^\ast \}))}^\ast$ then $f (x^\ast) = 0$.
\end{lemma}
	
\begin{proof}
Let $(x_i)_{i \in I} \subseteq \cone (E_2 \cup - (E_1 \setminus \{x^\ast\}))$ be such that $x^\ast = \lim_{i \in I} x_i$ (in the weak topology). Then for every $\varepsilon > 0$ there exists $i ^\ast \in I$ such that 
		
\begin{equation*}
\left|   f (x^\ast) -  f ( x_i) \right| < \varepsilon \, ,
\end{equation*}
		
\noindent for all $i \geq i^\ast$. 
		
Take $x_i = \sum_{k=1}^m \alpha_k e_k - \sum_{l = 1}^n \beta_l e'_l$, where $i \geq i^\ast$, $\alpha_1,\ldots,\alpha_m,\beta_1,\ldots,\beta_n > 0$, $e_1,\ldots,e_m \in E_2$, $e'_1,\ldots,e'_n \in E_1$. Then
		
\begin{equation*}
\left|  f (x^\ast) + f ( \sum_{l = 1}^n \beta_l e'_l) -  f ( \sum_{k=1}^m \alpha_k e_k ) \right| < \varepsilon \, .
\end{equation*}
		
Since $f (  \sum_{l = 1}^n \beta_l e'_l), -f ( \sum_{k=1}^m \alpha_k e_k ) \geq 0$ we have that  
\begin{equation*}
| f (x^\ast ) | < \varepsilon \, .
\end{equation*}
\noindent Since $\varepsilon > 0$ was arbitrarily chosen and $f ( x^\ast ) \geq 0$, we have that $f ( x^\ast ) = 0$.
\end{proof}
	
\begin{theorem}\label{thm:sep}
Let $X$ be a vector space, $E_1,E_2 \subseteq X$. Then the sets $\overline{\cone (E_1)}^\ast$ and $\overline{\cone (E_2)}^\ast$ are properly separable if and only if there exists $x^\ast \in E_1$ such that $x^\ast \notin \overline{\cone (E_2 \cup - (E_1 \setminus \{x^\ast\}))}^\ast$ or there exists $x^\ast \in E_2$ such that $x^\ast \notin \overline{\cone (E_1 \cup - (E_2 \setminus \{x^\ast\}))}^\ast$.
\end{theorem}
	
\begin{proof}
%We give the proof for the case $x^\ast \in E_1$. The proof of the other case goes analogously. Moreover, n
Notice that $X$ equipped with the weak topology is a locally convex topological vector space.
		
If: Suppose that  there exists $x^\ast \in E_1$ such that $x^\ast \notin \overline{\cone (E_2 \cup - (E_1 \setminus \{x^\ast\}))}^\ast$. Then $x^\ast$ and $\overline{\cone (E_2 \cup - (E_1 \setminus \{x^\ast\}))}^\ast$ are strongly separable (see Corollary 5.84 on. p. 209 in \citet{AliprantisBorder2006}), that is, there exists $g \in X'$ such that $g (x^\ast) > 0$ and $g (x) \leq 0$, for all $x \in \overline{\cone (E_2 \cup - (E_1 \setminus \{x^\ast\}))}^\ast$.
		
Then for each $e \in E_1 \setminus \{x^\ast\}$ it holds that $-e \in \overline{\cone (E_2 \cup - (E_1 \setminus \{x^\ast\}))}^\ast$, hence $g (e) \geq 0$. Moreover, for each $e \in E_2$ it holds that $e \in \overline{\cone (E_2 \cup - (E_1 \setminus \{x^\ast\}))}^\ast$, hence $g (e) \leq 0$. Therefore, $g (e) \geq 0 \geq g (e')$, for all $e \in E_1$, $e' \in E_2$, and $g (x^\ast) > 0$.
		
\bigskip
		
Only if:   Let $g \in X'$ be such that $g (e) \geq c \geq g (e')$, for all $e \in \overline{\cone (E_1)}^\ast$, $e' \in \overline{\cone (E_2)}^\ast$, and w.l.o.g. we can assume that there exists $x^\ast \in \overline{\cone (E_1)}^\ast$ such that $g (x^\ast) > c$. Then by that $0 \in \cone (E_1) \cap \cone (E_2)$ we have that $c=0$, that is, $g (e) \geq 0 \geq g (e')$, for all $e \in \overline{\cone (E_1)}^\ast$, $e' \in \overline{\cone (E_2)}^\ast$. 
		
Indirectly assume that for every $e \in E_1$ it holds that $e \in \overline{\cone (E_2 \cup - (E_1 \setminus \{e\}))}^\ast$. Then by Lemma \ref{lem:appkell} $g (e) = 0$, for all $e \in E_1$. Therefore, $g (x) = 0$, for all $x \in \overline{\cone (E_1)}^\ast$, which is a contradiction.
\end{proof}

\end{document}

\maketitle

\section{The Charge Setting}

Given a set $X$, a field\footnote{\ The mathematical literature uses two synonyms for the same concept: field and algebra. We elected in this paper to use the term field for this concept.} $\mathcal{A}$ is a collection of subsets of $X$ that is closed under complements and under finite intersections (hence also under finite unions) of sets, and contains $X$ itself.
	
If $\mathcal{A}$ is a field on $X$ then the pair $(X,\mathcal{A})$ is called a \emph{chargeable space}. A set function $\mu \colon \mathcal{A} \to \mathbb{R}$ on a chargeable space is a \emph{charge} on $(X,\mathcal{A})$ if $\mu$ is additive; it is a \textit{probability charge} if in addition it is non-negative and $\mu (X) = 1$.  Let $\pba (X,\mathcal{A})$ denote the collection of probability charges over the chargeable space $(X,\mathcal{A})$.
	
%If $\mathcal{A}$ is a $\sigma$-field the pair $(X,\mathcal{A})$ is called a \emph{measurable space}. A charge $\mu$ over $\mathcal{A}$ that is $\sigma$-additive is a \emph{measure}; it is a \emph{probability measure} if it is non-negative and $\mu (X) = 1$. Let $\Delta (X,\mathcal{A})$ denote the set of probability measures. Moreover, $\ba (X,\mathcal{A})$ and $\ca (X,\mathcal{A})$ will denote, respectively, the set of bounded charges and the set of bounded measures. Note in particular that each measure is a charge, but not the converse.
	
Let $B(X,\mathcal{A})$ denote the collection of all uniform limits of finite linear combinations of characteristic functions of sets in $\mathcal{A}$. This is the set of functions over $X$ which are integrable by every probability charge from $\pba(X,\mathcal{A})$. We will refer to the elements of $B(X,\mathcal{A})$ by the term \emph{$\mathcal{A}$-integrable functions}.

In this charge approach denote by $B^\ast(X,\mathcal{A})$ the dual of (the set of linear functionals on) $B(X,\mathcal{A})$, noting that $\ba (X,\mathcal{A}) = B^\ast (X,\mathcal{A})$ (see, e.g., \citet{DunfordSchwartz1958} 1 Theorem p. 258), where $\ba (X,\A)$ denotes the set of bounded charges over $(X,\A)$. Therefore, we consider the dual pair $(B(X,\A),\ba(X,\A))$.

Let $\ba (X,\mathcal{A})$ be equipped with the weak* topology by $B(X,\mathcal{A})$. For any $\mu \in \ba (X,\mathcal{A})$, $\varepsilon > 0$, and $F \subseteq B (X,\mathcal{A})$ such that $| F | < \infty$, denote by 	
\begin{equation*}
\mathcal{O}^\ast (\mu,F,\varepsilon) = \left\{ \nu \in \ba (X,\mathcal{A}) \colon \left| \int f \dd \nu - \int f \dd \mu \right| < \varepsilon, \  \forall f \in F \right\},
\end{equation*} 
	
\noindent a neighbourhood of $\mu$. The collection of such neighbourhoods forms a subbase for the weak* topology on $\ba (X,\mathcal{A})$. $\overline{A}^\ast$ will denote the weak* closure of a set $A \subseteq \ba (X,\mathcal{A})$.
		
\bigskip

We introduce in Definition \ref{def:pts2} the model of type space for this paper. This is a generalisation of the notion of type space appearing in \citet{Samet1998,HeifetzSamet1998}, and is identical to the notion in \citet{HellmanPinter2022}. 
	
\begin{definition}\label{def:pts2}
A \emph{type space} is a tuple $((\Omega,\mathcal{A}),\{(\Omega,\mathcal{M}_i)\}_{i \in N},\{t_i\}_{i \in N})$, where 
		
\begin{itemize}
\item $\O$ is a set of states of the world (state space),
			
\item $\mathcal{M}_i$ is the class of events player $i$ knows; this is called player $i$'s knowledge field and it is a field over $\O$, for each $i \in N$,
			
\item $\mathcal{A}$ is the field of all considered events, satisfying $\mathcal{M}_i \subseteq \mathcal{A}$, for all $i \in N$,
			
\item for each player $i \in N$, $t_i \colon \Omega \times \mathcal{A} \to [0,1]$ is a type function for player $i$, satisfying
			
\begin{enumerate}
\item $t_i(\o, \cdot) \in \pba (\O,\mathcal{A})$ is the belief of player $i$ at the state of world $\o$, for all $\o \in \O$, 
				
\item $t_i(\cdot, E)$ is $\mathcal{M}_i$-integrable, that is, player $i$ knows her own belief about the event $E$, for all $E \in \mathcal{A}$,
				
\item for each $E \in \mathcal{M}_i$ and each $\o \in E$, $t_i(\o,E) = 1$, that is, player $i$ believes with probability $1$ the events she knows.
\end{enumerate}
			
\end{itemize}
\end{definition}
	
For a detailed interpretation of this type space see \citet{HellmanPinter2022}.
	
Comparing the above notion of type spaces to the notions of type spaces which appear under different names in the literature ( see e.g. \citet{MertensZamir1985,BrandenburgerDekel1993,Heifetz1993,HeifetzSamet1998} among others), the main difference between this notion and the other notions is that in a type space as defined here a state of the world does not necessarily determine the belief hierarchies of the players. In other words, while in a conventional type space a state of the world determines the players' belief hierarchies, in a type space as defined in Definition \ref{def:pts2} it might be impossible to define higher order beliefs. 

\begin{example}\label{ex:pts}[A type space without higher order beliefs.]
Consider the type space $((\Omega,\mathcal{A}),\{(\Omega,\mathcal{M}_i)\}_{i \in N},\{t_i\}_{i \in N})$, where 
		
\begin{itemize}
\item $S = \{s_1,s_2\}$ is the parameter set (in a Bayesian game setting one may regard this as the set of payoff relevant parameters),
			
\item $N = \{1,2\}$ is the player set,
			
\item $\O = S \times \{t_1\} \times \mathbb{N}$ is the state space. Therefore, player $1$ has only one type $t_1$, while player $2$'s type set is $\mathbb{N}$,
			
\item $\mathcal{N} = \{A \subseteq \mathbb{N} \colon| A | < \infty \text{ or } | A^\complement | < \infty \}$ is the field of the finite and co-finite subsets of the natural numbers,
			
\item $\mathcal{M}_1 = \{S\} \otimes \{t_1\} \otimes \{\mathbb{N}\}$, where $\otimes$ denotes the generated product field. Meaning, player $1$ knows only her single type ($t_1$), hence her knowledge field $\mathcal{M}_1$ is the minimal field over $\O$ consisting of only the empty set and $\O$,
			
\item $\mathcal{M}_2 = \{S\} \otimes \{t_1\} \otimes \mathcal{N}$, meaning player $2$ knows only her types ($\mathcal{N}$), hence her knowledge field $\mathcal{M}_2$ is the field over $\O$ generated by $\mathcal{N}$ through the coordinate projection from $\O$ to $\mathbb{N}$,
			
\item $\mathcal{A} = \mathcal{P} (S) \otimes \{t_1\} \otimes \mathcal{N}$, the field $\mathcal{A}$ is generated by all subsets of the parameter set and by $\mathcal{N}$ through the corresponding coordinate projections respectively. Therefore, $\mathcal{M}_1 \subseteq \mathcal{A}$ and $\mathcal{M}_2 \subseteq \mathcal{A}$,
			
\item $t_1 ((s,t_1,n),\cdot) = \pi  \times \delta_{t_1} \times \mu$, where $\pi \in \Delta(S,\mathcal{P} (S))$ such that $\pi (\{s_1\}) = \pi (\{s_2\}) = \frac{1}{2}$, that is, $\pi$ is the uniform distribution over $S$, $\delta_x$ is the Dirac measure at point $x$, and $\mu \in \pba(\mathbb{N},\mathcal{N})$ is the uniform distribution, that is, $\mu (E) = 0$ if $|E| < \infty$, and $\mu (E) =1$ otherwise, $E \in \mathcal{N}$. In other words, player $1$'s belief is the uniform distribution over $\mathcal{A}$, in the sense that she thinks every state in $\O$ has the same probability,
			
\item $t_2 ((s,t_1,n),\cdot) = \pi_n \times \delta_{t_1} \times \delta_{n}$, where $\pi_n \in \pba (S,\mathcal{P} (S))$ is such that 
\begin{equation*}
\pi_n (\{s_1\}) = \left\{
\begin{array}{ll}
0 & \text{if } n \text{ is odd}, \\
\frac{1}{2 \, 2^{n/2}} & \text{if } n \text{ is even}. \\
\end{array}
\right.
\end{equation*}
\end{itemize}
		
Take the state $\o = (s_1,t_1,5)$. Then the first order belief of player $1$ for the event $E = \{s_1\} \times \{t_1\} \times \mathbb{N}$ (the event of that $s_1$ occurs) is $t_1 (\o,E) = \frac{1}{2}$.
		
The second order belief of player $1$ is not defined here since the event that player $2$ believes that $s_1$ occurs with positive probability, namely $\{s_1 \} \times \{t_1\} \times \{ n \in \mathbb{N} \colon \pi_n (\{s_1\}) > 0\}$, is not an element of $\mathcal{A}$, that is, player $1$ cannot form a belief about this event. 
		
Similarly, player $1$ cannot form a belief about the event that player $2$ believes that $s_2$ occurs with probability $1$, as $\{s_1 \} \times \{t_1\} \times \{ n \in \mathbb{N} \colon \pi_n (\{s_2\}) = 1\} \notin \mathcal{A}$.
\end{example}
	
Clearly, for every finite state space (that is, when $\Omega$ is finite), belief hierarchies are well defined. Moreover, as illustrated by Example \ref{ex:pts}, in the case of infinite type spaces, higher-order beliefs play no role in our analysis or results.

The notions of bet, acceptable bet, and strong consistency of beliefs are essentially the same as in the probability measure setting (see the definitions in the paper); the differences are only minor technical ones, and we therefore do not restate them here.

We now state the infinite-state-space generalization, in the probability measure setting, of the strong consistency of beliefs–acceptable bet epistemic–behavioural dual characterisation of \citet{Morris1994}. Since its proof follows that of the theorem in the probability measure setting (see the paper), we omit it.
	
%\begin{definition}\label{def:ccp3}
%The players' beliefs in a type space $((\Omega,\mathcal{A}),\{(\Omega,\mathcal{M}_i)\}_{i \in N},\{t_i\}_{i \in N})$ are \emph{strongly consistent} if for every $I \subseteq N$, $|I| < \infty$, every player $i^\ast \in I$ and every type of player $i^\ast$ $t_{i^\ast} (\o^\ast , \cdot)$, and every $\o^\ast \in \O$,  
%\begin{equation*}
%t_{i^\ast} (\o^\ast,\cdot) \in \overline{\cone (\{t_i (\o, \cdot) \colon \o \in \O \}) - \{t_{i^\ast} (\o,\cdot) \colon \o \in \O \text{ s.t. } t_{i^\ast} (\o,\cdot) \neq t_{i^\ast} (\o^\ast,\cdot)  \})}^\ast \, ,
%\end{equation*}
		
%\noindent for all $i \in I \setminus \{i^\ast\}$.
%\end{definition}	

%The following theorem is an infinite state space generalisation of the strong consistency of beliefs -- acceptable bet by \citet{Morris1994} to the model with charges. It states that the players' beliefs in a type space are strongly consistent if and only if there does not exist an acceptable bet in the type space. For the proof see the proof of Theorem \ref{thm:notrade3}.
	
\begin{theorem}
Let $T = ((\O,\A),(\O,\M_i)_{i \in N},(t_i)_{i \in N})$ be a type space. Then only one of the following two statements can obtain:
		
\begin{itemize}
\item The players' beliefs in the type space are strongly consistent.
			
\item There exists an acceptable bet in the type space.
\end{itemize}
\end{theorem}
	
\bigskip
	
In what follows, we present three examples that illustrate the nature of strong consistency of players’ beliefs (Definition \ref{def:cp3}). All three are inspired by the example in Figure 2, p. 152 of \citet{Feinberg2000}. Compared with the corresponding examples in the probability measure setting, the examples here are simpler, more subtle, and more elegant.
	
\begin{example}\label{exp:cFeinbergv1}
Consider the type space in Figure \ref{fig:cFeinbergv1}, taken from \citet{Feinberg2000}.
		
\begin{figure}[h!]
\begin{tabular}{l | m{1mm} m{0.1mm} m{1mm} m{0.1mm} m{1mm} m{0.1mm} m{1mm} m{0.1mm} m{1mm} m{0.1mm} m{1mm} m{0.1mm} m{1mm} m{0.1mm} m{1mm} m{0.1mm} m{1mm} m{0.1mm} m{1mm} m{0.1mm} m{1mm} m{0.1mm} m{1mm} m{0.1mm} m{1mm} m{0.1mm} m{1mm} m{0.1mm} m{1mm} m{0.1mm} m{1mm} m{0.1mm} m{1mm}}
\text{Anne} & 1   &  \vline  &  $\frac{2}{3}$ &  &  $\frac{1}{3}$ & \vline & $\frac{1}{2}$ &  &  $\frac{1}{2}$ & \vline & $\frac{2}{3}$ &  &  $\frac{1}{3}$ & \vline & $\frac{1}{2}$ &  &  $\frac{1}{2}$ & \vline & $\frac{1}{2}$ &  & $\frac{1}{2}$ & \vline & $\frac{1}{2}$ &  & $\frac{1}{2}$ & \vline & $\frac{2}{3}$ &  & $\frac{1}{3}$ & \vline & $\ldots$ \\
\hline
& $\o_1$ & \vline & $\o_2$ & \vline & $\o_3$ & \vline & $\o_4$ & \vline & $\o_5$ & \vline & $\o_6$ & \vline & $\o_7$ & \vline & $\o_8$ & \vline & $\o_9$ & \vline & $\o_{10}$ & \vline & $\o_{11}$ & \vline & $\o_{12}$ & \vline & $\o_{13}$ & \vline & $\o_{14}$ & \vline & $\o_{15}$ & \\ 
\hline
\text{Ben} & $\frac{1}{2}$ &  &  $\frac{1}{2}$ & \vline  & $\frac{1}{2}$  & &  $\frac{1}{2}$ & \vline &
$\frac{1}{2}$  &  &  $\frac{1}{2}$ & \vline & $\frac{1}{2}$  &  &  $\frac{1}{2}$ & \vline & $\frac{1}{2}$  &  &  $\frac{1}{2}$ & \vline & $\frac{1}{2}$  &  &  $\frac{1}{2}$ & \vline & $\frac{1}{2}$  &  &  $\frac{1}{2}$ & \vline& $\frac{1}{2}$  & & $\ldots$       \\ 
\end{tabular}
\caption{The type space of the example.}\label{fig:cFeinbergv1}
\end{figure}
		
Formally, the type space $((\Omega,\mathcal{A}),\{(\Omega,\mathcal{M}_i)\}_{i \in N},\{t_i\}_{i \in N})$ is as follows:
		
\begin{itemize}
\item $N = \{\text{Anne},\text{Ben}\}$,
			
\item $\O = \N$,
			
\item $\A = \{ A \in \mathcal{P} (\O) \colon | A | < \infty \text{ or } | A^\complement | < \infty \}$, that is, $\A$ is the so called finite, co-finite algebra,
			
\item $\M_{\text{Anne}}$ is the field generated by the partition $\{\{\o_1\},\{\o_2,\o_3\},\{\o_4,\o_5\},\ldots\}$, and $\M_{\text{Ben}}$ is the field generated by the partition $\{\{\o_1,\o_2\},\{\o_3,\o_4\},\{\o_5,\o_6\},\ldots\}$,
			
\item $t_{\text{Anne}} (\o_n,\cdot)= 
\begin{cases}
\delta_{\o_1} & \text{if } n=1, \\
\frac{2}{3} \delta_{\o_n} + \frac{1}{3} \delta_{\o_{n+1}} & \text{if } n = 2^{m + 1} - 2,\ \text{for some } m \in \N, \\
\frac{2}{3} \delta_{\o_{n-1}} + \frac{1}{3} \delta_{\o_{n}} & \text{if } n = 2^{m + 1} - 1,\ \text{for some } m \in \N, \\
\frac{1}{2} \delta_{\o_{n}} + \frac{1}{2} \delta_{\o_{n + 1}} & \text{if } n \text{ is even, and } n  \neq 2^{m + 1} - 2,\ \text{for every } m \in \N, \\
\frac{1}{2} \delta_{\o_{n-1}} + \frac{1}{2} \delta_{\o_{n}} & \text{if } n \text{ is odd, and } n  \neq 2^{m + 1} - 1,\ \text{for every } m \in \N , \\
\end{cases}$, and $t_{\text{Ben}} (\o_n, \cdot) = 
\begin{cases}
\frac{1}{2} \delta_{\o_{n}} + \frac{1}{2} \delta_{\o_{n + 1}} & \text{if } n \text{ is odd}, \\
\frac{1}{2} \delta_{\o_{n - 1}} + \frac{1}{2} \delta_{\o_{n}} & \text{if } n \text{ is even}. \\
\end{cases}$
\end{itemize}
		
In this type space the players' beliefs are consistent, and there is only one witness of this consistency: the uniform distribution, that is, the probability charge

\begin{equation*}
\nu (A) = 
\begin{cases}
0 & \text{if } | A | < \infty \, , \\
1 & \text{otherwise}. 
\end{cases}
\end{equation*}
		
Consider $t_{\text{Anne}} (\o_1,\cdot)$. By Theorem \ref{thm:sep} if the point $t_{\text{Anne}} (\o_1,\cdot)$ and the cone $\overline{\cone (\{t_{\text{Ben}} (\o,\cdot) \colon \o \in \O\} - \{t_{\text{Anne}} (\o,\cdot) \colon \o \in \O \setminus \{\o_1\} \})}^\ast$ can be separated then $\Pi_{\text{Anne}}$ and $\Pi_{\text{Ben}}$ can be properly separated. We therefore need to check whether $\Pi_{\text{Anne}}$ and $\Pi_{\text{Ben}}$ can be properly separated or not.
		
Suppose that $f$ properly separates $\Pi_{\text{Anne}}$ and $\Pi_{\text{Ben}}$. Without loss of generality we can assume that $f (\o_1) = 1$. Then by that $\int f \dd t_{\text{Ben}} (\o_1,\cdot) \leq 0$ we have that $f (\o_2) \leq -1$. Then by that $\int f \dd t_{\text{Anne}} (\o_2,\cdot) \geq 0$ we have that $f (\o_3) \geq 2$, and so on. Therefore, we get $f (\o_6) \geq 4$, $f (\o_{14}) \geq 8$, and so on, that is, $f$ is not bounded, hence it cannot be a separator; meaning $\Pi_{\text{Anne}}$ and $\Pi_{\text{Ben}}$ cannot be properly separated. Therefore, 
\begin{equation*}
t_{\text{Anne}} (\o_1,\cdot) \in \overline{\cone (\{t_{\text{Ben}} (\o,\cdot) \colon \o \in \O\} - \{t_{\text{Anne}} (\o,\cdot) \colon \o \in \O \setminus \{\o_1\}\})}^\ast \, .
\end{equation*}
		
By very similar reasoning one can show that the same  holds for every type of each player, that is, the players' beliefs are strongly consistent. 
\end{example}
	
\begin{example}\label{exp:cFeinbergv2}
Consider the type space in Figure \ref{fig:cFeinbergv2}, which is a variant of a type space considered by \citet{Feinberg2000}.
		
\begin{figure}[h!]
\begin{tabular}{l | m{1mm} m{0.1mm} m{1mm} m{0.1mm} m{1mm} m{0.1mm} m{1mm} m{0.1mm} m{1mm} m{0.1mm} m{1mm} m{0.1mm} m{1mm} m{0.1mm} m{1mm} m{0.1mm} m{1mm} m{0.1mm} m{1mm} m{0.1mm} m{1mm} m{0.1mm} m{1mm} m{0.1mm} m{1mm} m{0.1mm} m{1mm} m{0.1mm} m{1mm} m{0.1mm} m{1mm} m{0.1mm} m{1mm}}
\text{Anne} & 1   &  \vline  &  $\frac{1}{2}$ &  &  $\frac{1}{2}$ & \vline & $\frac{1}{2}$ &  &  $\frac{1}{2}$ & \vline & $\frac{1}{2}$ &  &  $\frac{1}{2}$ & \vline & $\frac{1}{2}$ &  &  $\frac{1}{2}$ & \vline & $\frac{1}{2}$ &  & $\frac{1}{2}$ & \vline & $\frac{1}{2}$ &  & $\frac{1}{2}$ & \vline & $\frac{1}{2}$ &  & $\frac{1}{2}$ & \vline & $\ldots$ \\
\hline
& $\o_1$ & \vline & $\o_2$ & \vline & $\o_3$ & \vline & $\o_4$ & \vline & $\o_5$ & \vline & $\o_6$ & \vline & $\o_7$ & \vline & $\o_8$ & \vline & $\o_9$ & \vline & $\o_{10}$ & \vline & $\o_{11}$ & \vline & $\o_{12}$ & \vline & $\o_{13}$ & \vline & $\o_{14}$ & \vline & $\o_{15}$ & \\ 
\hline
\text{Ben} & $\frac{1}{2}$ &  &  $\frac{1}{2}$ & \vline  & $\frac{1}{2}$  & &  $\frac{1}{2}$ & \vline & $\frac{1}{2}$  &  &  $\frac{1}{2}$ & \vline & $\frac{1}{2}$  &  &  $\frac{1}{2}$ & \vline & $\frac{1}{2}$  &  &  $\frac{1}{2}$ & \vline & $\frac{1}{2}$  &  &  $\frac{1}{2}$ & \vline & $\frac{1}{2}$  &  &  $\frac{1}{2}$ & \vline& $\frac{1}{2}$  & & $\ldots$       \\ 
\end{tabular}
\caption{The type space of the example.}\label{fig:cFeinbergv2}
\end{figure}
		
Formally, the type space $((\Omega,\mathcal{A}),\{(\Omega,\mathcal{M}_i)\}_{i \in N},\{t_i\}_{i \in N})$ is as follows:
		
\begin{itemize}
\item $N = \{\text{Anne},\text{Ben}\}$,
			
\item $\O = \N$,
			
\item $\A = \{ A \in \mathcal{P} (\O) \colon | A | < \infty \text{ or } | A^\complement | < \infty \}$, that is, $\A$ is the so called finite, co-finite algebra,
			
\item $\M_{\text{Anne}}$ is the field generated by the partition $\{\{\o_1\},\{\o_2,\o_3\},\{\o_4,\o_5\},\ldots\}$, and $\M_{\text{Ben}}$ is the field generated by the partition $\{\{\o_1,\o_2\},\{\o_3,\o_4\},\{\o_5,\o_6\},\ldots\}$,
			
\item $t_{\text{Anne}} (\o_n,\cdot)= 
\begin{cases}
\delta_{\o_1} & \text{if } n=1, \\
\frac{1}{2} \delta_{\o_{n}} + \frac{1}{2} \delta_{\o_{n + 1}} & \text{if } n \text{ is even}, \\
\frac{1}{2} \delta_{\o_{n-1}} + \frac{1}{2} \delta_{\o_{n}} & \text{if } n \text{ is odd, and } n  \neq 1, \\
\end{cases}$, and $t_{\text{Ben}} (\o_n, \cdot) = 
\begin{cases}
\frac{1}{2} \delta_{\o_{n}} + \frac{1}{2} \delta_{\o_{n + 1}} & \text{if } n \text{ is odd}, \\
\frac{1}{2} \delta_{\o_{n - 1}} + \frac{1}{2} \delta_{\o_{n}} & \text{if } n \text{ is even}. \\
\end{cases}$
\end{itemize}
		
In this type space, as in the previous example, the players beliefs are consistent, and there is only one witness of this consistency: the uniform distribution.
		
As in the previous example consider $t_{\text{Anne}} (\o_1,\cdot)$. By Theorem \ref{thm:sep} if the point $t_{\text{Anne}} (\o_1,\cdot)$ and the cone $\overline{\cone (\{t_{\text{Ben}} (\o,\cdot) \colon \o \in \O\} - \{t_{\text{Anne}} (\o,\cdot) \colon \o \in \O \setminus \{\o_1\} \})}^\ast$ can be separated then $\Pi_{\text{Anne}}$ and $\Pi_{\text{Ben}}$ can be properly separated. We therefore need to check whether $\Pi_{\text{Anne}}$ and $\Pi_{\text{Ben}}$ can be properly separated or not.
		
Suppose that $f$ properly separates $\Pi_{\text{Anne}}$ and $\Pi_{\text{Ben}}$. Then, without loss of generality we can assume that $f (\o_1) = 1$. Then by that $\int f \dd t_{\text{Ben}} (\o_1,\cdot) \leq 0$ we have that $f (\o_2) \leq -1$. Then by that $\int f \dd t_{\text{Anne}} (\o_2,\cdot) \geq 0$ we have that $f (\o_3) \geq 1$, and so on. Therefore, we get that $f (\o_n) = 
\begin{cases}
\geq 1 & \text{if } n \text{ is odd}, \\
\leq -1 & \text{if } n \text{ is even}.
\end{cases}$
		
However, then even if $f$ is bounded, it is not $\A$-integrable; meaning $\Pi_{\text{Anne}}$ and $\Pi_{\text{Ben}}$ cannot be properly separated. Therefore, 
\begin{equation*}
t_{\text{Anne}} (\o_1,\cdot) \in \overline{\cone (\{t_{\text{Ben}} (\o,\cdot) \colon \o \in \O\} - \{t_{\text{Anne}} (\o,\cdot) \colon \o \in \O \setminus \{\o_1\}\})}^\ast \, .
\end{equation*}
		
By very similar reasoning one can show that the same  holds for every type of each player, that is, the players' beliefs are strongly consistent. 
\end{example}
	
\begin{example}\label{exp:cFeinbergv3}
Consider the type space in Figure \ref{fig:cFeinbergv3}. This is again a variant of a type space considered by \citet{Feinberg2000}.
		
\begin{figure}[h!]
\begin{tabular}{l | m{1mm} m{0.1mm} m{1mm} m{0.1mm} m{1mm} m{0.1mm} m{1mm} m{0.1mm} m{1mm} m{0.1mm} m{1mm} m{0.1mm} m{1mm} m{0.1mm} m{1mm} m{0.1mm} m{1mm} m{0.1mm} m{1mm} m{0.1mm} m{1mm} m{0.1mm} m{1mm} m{0.1mm} m{1mm} m{0.1mm} m{1mm} m{0.1mm} m{1mm} m{0.1mm} m{1mm} m{0.1mm} m{1mm}}
\text{Anne} & 1   &  \vline  &  $\frac{1}{3}$ &  &  $\frac{2}{3}$ & \vline & $\frac{1}{2}$ &  &  $\frac{1}{2}$ & \vline & $\frac{1}{3}$ &  &  $\frac{2}{3}$ & \vline & $\frac{1}{2}$ &  &  $\frac{1}{2}$ & \vline & $\frac{1}{2}$ &  & $\frac{1}{2}$ & \vline & $\frac{1}{2}$ &  & $\frac{1}{2}$ & \vline & $\frac{1}{3}$ &  & $\frac{2}{3}$ & \vline & $\ldots$ \\
\hline
& $\o_1$ & \vline & $\o_2$ & \vline & $\o_3$ & \vline & $\o_4$ & \vline & $\o_5$ & \vline & $\o_6$ & \vline & $\o_7$ & \vline & $\o_8$ & \vline & $\o_9$ & \vline & $\o_{10}$ & \vline & $\o_{11}$ & \vline & $\o_{12}$ & \vline & $\o_{13}$ & \vline & $\o_{14}$ & \vline & $\o_{15}$ & \\ 
\hline
\text{Ben} & $\frac{1}{2}$ &  &  $\frac{1}{2}$ & \vline  & $\frac{1}{2}$  & &  $\frac{1}{2}$ & \vline & $\frac{1}{2}$  &  &  $\frac{1}{2}$ & \vline & $\frac{1}{2}$  &  &  $\frac{1}{2}$ & \vline & $\frac{1}{2}$  &  &  $\frac{1}{2}$ & \vline & $\frac{1}{2}$  &  &  $\frac{1}{2}$ & \vline & $\frac{1}{2}$  &  &  $\frac{1}{2}$ & \vline& $\frac{1}{2}$  & & $\ldots$ \\ 
\end{tabular}
\caption{The type space of the example.}\label{fig:cFeinbergv3}
\end{figure}
		
Formally, the type space $((\Omega,\mathcal{A}),\{(\Omega,\mathcal{M}_i)\}_{i \in N},\{t_i\}_{i \in N})$ is as follows:
		
\begin{itemize}
\item $N = \{\text{Anne},\text{Ben}\}$,
			
\item $\O = \N$,
			
\item $\A = \{ A \in \mathcal{P} (\O) \colon | A | < \infty \text{ or } | A^\complement | < \infty \}$, that is, $\A$ is the so called finite, co-finite algebra,
			
\item $\M_{\text{Anne}}$ is the field generated by the partition $\{\{\o_1\},\{\o_2,\o_3\},\{\o_4,\o_5\},\ldots\}$, and $\M_{\text{Ben}}$ is the field generated by the partition $\{\{\o_1,\o_2\},\{\o_3,\o_4\},\{\o_5,\o_6\},\ldots\}$,
			
\item $t_{\text{Anne}} (\o_n,\cdot)= 
\begin{cases}
\delta_{\o_1} & \text{if } n=1, \\
\frac{1}{3} \delta_{\o_n} + \frac{2}{3} \delta_{\o_{n+1}} & \text{if } n = 2^{m + 1} - 2,\ \text{for some } m \in \N, \\
\frac{1}{3} \delta_{\o_{n-1}} + \frac{2}{3} \delta_{\o_{n}} & \text{if } n = 2^{m + 1} - 1,\ \text{for some } m \in \N, \\
\frac{1}{2} \delta_{\o_{n}} + \frac{1}{2} \delta_{\o_{n + 1}} & \text{if } n \text{ is even, and } n  \neq 2^{m + 1} - 2,\ \text{for every } m \in \N, \\
\frac{1}{2} \delta_{\o_{n-1}} + \frac{1}{2} \delta_{\o_{n}} & \text{if } n \text{ is odd, and } n  \neq 2^{m + 1} - 1,\ \text{for every } m \in \N, \\
\end{cases}$, and $t_{\text{Ben}} (\o_n, \cdot) = 
\begin{cases}
\frac{1}{2} \delta_{\o_{n}} + \frac{1}{2} \delta_{\o_{n + 1}} & \text{if } n \text{ is odd}, \\
\frac{1}{2} \delta_{\o_{n - 1}} + \frac{1}{2} \delta_{\o_{n}} & \text{if } n \text{ is even}. \\
\end{cases}$
\end{itemize}
		
In this type space, as in the two previous examples, the players beliefs are consistent, and there is only one witness of this consistency: the uniform distribution. However, differently from the two previous examples here the players' beliefs are not strongly consistent.
		
As in the two previous examples we check whether $\Pi_{\text{Anne}}$ and $\Pi_{\text{Ben}}$ can be properly separated or not.
		
Take the function defined as follows $f (\o_1) = 1$, $f (\o_2) = -1$, $f (\o_3) = 1/2$, $f (\o_4) = -1/2$, $f (\o_5) = 1/2$, $f (\o_6) = -1/2$, $f(\o_7) = 1/4$, $f (\o_8) = -1/4$, and so on. First of all, notice that $f \in B(\O,\A)$, that is, it is $\A$-integrable. Moreover, it is easy to check that $\int f \dd t_{\text{Anne}} (\o_1,\cdot) = 1$, $\int f \dd t_{\text{Ben}} (\o_1,\cdot) = 0$, and $\int f \dd t_i (\o_n,\cdot) = 0$, for all $i \in \{\text{Anne},\text{Ben}\}$, and $n > 1$. Therefore, $f$ properly separates the sets $\Pi_{\text{Anne}}$ and $\Pi_{\text{Ben}}$. Then by Theorem \ref{thm:sep} we have that there exist $i^\ast \in \{\text{Anne},\text{Ben}\}$ and $\o^\ast \in \O$ such that  
\begin{equation*}
t_{i^\ast} (\o^\ast,\cdot) \notin \overline{\cone (\{t_{-i^\ast} (\o,\cdot) \colon \o \in \O\} - \{t_{i^\ast} (\o,\cdot) \colon \o \in \O \setminus \{\o^\ast\}\})}^\ast \, .
\end{equation*}
		
Therefore, the players' beliefs are not strongly consistent.
\end{example}
	
Examples \ref{exp:cFeinbergv1}–\ref{exp:cFeinbergv3} demonstrate that the strong consistency of players’ beliefs cannot be captured by the intersection of their prior sets. Even a set of probability distributions cannot, in general, witness the strong consistency of players’ beliefs.

The three type spaces differ only in the players’ types, which are nonetheless similar in that they share the same unique common prior. Yet, in Examples \ref{exp:cFeinbergv1} and \ref{exp:cFeinbergv2}, the players’ beliefs are strongly consistent, whereas in Example \ref{exp:cFeinbergv3} they are not. This shows that strong consistency cannot be characterized by probability distributions, and thus that neither a literal ex-ante stage nor the interpretation of a common prior as a real probability distribution is applicable in this infinite setting.

%Notice that in each of the Examples \ref{exp:cFeinbergv1}--\ref{exp:cFeinbergv3} the players' beliefs are (universally) consistent and there is a unique witness of the consistency: the uniform distribution. However, in \ref{exp:cFeinbergv1} and \ref{exp:cFeinbergv2} the players' beliefs are strongly consistent, while in Example \ref{exp:cFeinbergv3} those are not. 

%\begin{lemma}\label{lem:important0}
%Let $((\Omega,\mathcal{A}),\{(\Omega,\mathcal{M}_i)\}_{i \in N},\{t_i\}_{i \in N})$ be a type space. If for every $I \subseteq N$, $|I| < \infty$, $i^\ast \in I$, and $\o^\ast \in \O$ there exists $P \in \bigcap_{i \in I} \Pi_i$ such that $\inf_{\o^\ast \in E \in \M_{i^\ast}} P (E) > 0$, then the players' beliefs in the type space are strongly consistent. 
%\end{lemma}
	
%Notice that in the finite type space setting the measure and the charge approaches are the same, hence Lemma \ref{lem:important0} applies here.

%\bibliographystyle{/Users/miklospinter/Dropbox/Munka/Styles/ANOR/spbasic}
%\bibliography{/Users/miklospinter/Dropbox/Munka/Hivatkozasok/referencesPMP.bib}
	
%\bibliographystyle{/home/pmiklos/Dropbox/Munka/Styles/ANOR/spbasic}
%\bibliography{/home/pmiklos/Dropbox/Munka/Hivatkozasok/referencesPMP.bib}